\documentclass[a4paper,UKenglish,cleveref, autoref, thm-restate]{lipics-v2021}

\pdfoutput=1 
\hideLIPIcs  

\title{Mechanical proving with Walnut for squares and cubes in partial words} 

\titlerunning{Walnut for partial words} 

\author{John Machacek}{Department of Mathematics, University of Oregon, USA \and \url{https://sites.google.com/view/jmachacek/} }{johnmach@uoregon.edu}{}{}

\authorrunning{J. Machacek} 

\Copyright{John Machacek} 



\category{} 

\relatedversion{} 

\funding{NSF DMS-2039316}

\acknowledgements{This work benefited from the use of the CrySP RIPPLE Facility at the
University of Waterloo. Thanks to Ian Goldberg for allowing Jeffrey
Shallit to run a computation on this machine.}

\nolinenumbers 

\usepackage{diagbox}
\usepackage{amsmath}
\usepackage{amssymb}
\usepackage{amsthm}
\usepackage{mathrsfs}

\begin{document}

\maketitle

\begin{abstract}
Walnut is a software that can prove theorems in combinatorics on words about automatic sequences.
We are able to apply this software to both prove new results as well as reprove some old results on avoiding squares and cubes in partial words.
We also define the notion of an antisquare in a partial word and begin the study of binary partial words which contain only a fixed number of distinct squares and antisquares.
\end{abstract}

\section{Introduction}

Our focus is on repetitions in partial words and the use of the software called Walnut\footnote{We have used the version of Walnut available at \url{https://github.com/DistortedLight/Walnut}.}~\cite{Walnut} to give automated proofs in situations where automatic sequences can be used.
Partial words are generalizations of usual words that make use of an additional ``wildcard'' character which matches all other characters.
Walnut has been used to give alternative proofs of previously known results and has also been used to produce new theorems in combinatorics on words (see e.g.,~\cite{fib}).
To our knowledge this work is the first use of Walnut with partial words and contains proofs of both new results as well as previously known ones.
We will define the notions which are most central to our work, but familiarity with some standard terms and ideas from combinatorics on words~\cite{Loth} is assumed.
An introduction to fundamental concepts on partial words can be found in~\cite{FBS}
This paper is the full version of the conference paper presented at CPM 2022~\cite{CPM}.

A \emph{square} or \emph{cube} is a word of the form $xx$ or $xxx$, respectively, for a nonempty word $x$.
For example, $(010)^2 = 010010$ is a square and $(110)^3 = 110110110$ is a cube.
We will consider words that avoid squares as well as those which avoid cubes.
This means words that do not have a factor (i.e., contiguous substring) which is a square or cube respectively.
A \emph{squarefree} word is a word which avoids squares and a \emph{cubefree} word is a word which avoid cubes.
A \emph{morphism} is a map $\psi: \Sigma^* \to \Delta^*$ between words over two alphabets such that $\psi(xy) = \psi(x)\psi(y)$ for all $x,y \in \Sigma^*$.
We will make frequent use of morphisms to find words with a desired property.

A classic problem in combinatorics on words is constructing words avoiding squares, cubes, and other types of repetitions.
Thue~\cite{thue, Berstel} was able to construct an infinite cubefree binary word and an infinite squarefree ternary word each of which can be obtained as the fixed point of a morphism.
We let 
\[{\bf tm} = 01101001100101101001011001101001\cdots\]
denote the \emph{Thue-Morse} word which is the fixed point of the morphism $0 \mapsto 01$ and $1 \mapsto 10$ which begins with $0$.
We also let
\[{\bf vtm} = 012021012102012021020121\cdots\]
denote the fixed point of the morphism $0 \mapsto 012$, $1 \mapsto 02$, and $2 \mapsto 1$ which is sometimes called the \emph{ternary Thue-Morse word} (see e.g.,~\cite{binomial}) or a \emph{variant of the Thue-Morse word} (see~\cite{Fox}).
It is the case that ${\bf tm}$ is cubefree and ${\bf vtm}$ is squarefree.
We will make use of both of these words in constructions later.

A \emph{partial word} is a word which can use a special character $\diamond$ which is called a \emph{hole} or \emph{wildcard}.
For partial words a square or cube is a partial word ${\bf w}$ which is \emph{contained} in a square ${\bf u}=xx$ or a cube ${\bf u}=xxx$ respectively in the sense the ${\bf w}[i] = {\bf u}[i]$ whenever ${\bf w}[i] \neq \diamond$.
In this case we write ${\bf w} \subset {\bf u}$.
If the exists ${\bf w}$ such that ${\bf u} \subset {\bf w}$ and ${\bf v} \subset {\bf w}$ we say that ${\bf u}$ and ${\bf v}$ are \emph{compatible} and write ${\bf u} \uparrow {\bf v}$. 

The \emph{order} of the square or cube is the length of $x$, which we denote by $|x|$ .
For example, $01101\diamond 011$ is a partial word which is a cube of order $3$ since it is contained in $(011)^3 = 011011011$.
It is clear the presence of holes makes it more difficult to avoid squares or cubes.
All words are partial words with no holes.
When we wish to emphasize that a (partial) word has no holes we will refer to the word as a \emph{full word}.

\section{Squares, antisquares, cubes, and first-order logic}
\label{sec:defs}

In this section we define what squares, antisquares, and cubes are in terms of first-order logic.
This allows for seamless use with Walnut and highlights some differences between full words and partial words.
All variables we quantify over are taken from the nonnegative integers unless specified otherwise.

For a full word ${\bf w}$ containing a  \emph{square} means 
\[\exists j \exists (n>0)  \forall i, (i < n) \implies ({\bf w}[j+i] = {\bf w}[j+n+i])\]
while for a partial word it means
\[\exists j  \exists (n > 0)  \forall i, (i < n) \implies \left(({\bf w}[j+i] = {\bf w}[j+n+i]) \vee ({\bf w}[j+i]=\diamond) \vee ({\bf w}[j+n+i]=\diamond)\right)\]
both of which can be expressed in first-order logic.
We see that the expression for partial words contains more clauses.
A value of $n$ for which the above is made true is called the \emph{order} of the square.

A partial word ${\bf w}$ contains an \emph{antisquare} provided
\[\exists j  \exists (n > 0)  \forall i, (i < n) \implies \left(({\bf w}[j+i] \neq {\bf w}[j+n+i]) \wedge ({\bf w}[j+i]\neq\diamond) \wedge ({\bf w}[j+n+i]\neq\diamond)\right)\]
which is consistent with what was considered for binary words in~\cite{anti} and differs from the notion of an \emph{anti-power} studied in~\cite{otheranti}.
We believe this to be a natural definition of an antisquare for a partial word given how it comes from negating the latter half of the implication in the logical expression for a square.
We see that replacing a letter with a hole can create a square, and dually it can remove the presence of an antisquare.
A factor which is an antisquare cannot contain any holes, but since we will consider squares and antisquares together in partial words holes will still play a crucial role.

Lastly, we say a partial word ${\bf w}$  contains a \emph{cube} if
\begin{align*}
\exists j  \exists (n > 0)  \forall i, (i < n) \implies &\Big(\big(({\bf w}[j+i] = {\bf w}[j+n+i]) \wedge ({\bf w}[j+n+i] = {\bf w}[j+2n+i])\big) \\
	&\vee \big(({\bf w}[j+i]=\diamond) \wedge ({\bf w}[j+n+i]={\bf w}[j+2n+i])\big) \\
	&\vee  \big(({\bf w}[j+n+i]=\diamond) \wedge ({\bf w}[j+i]={\bf w}[j+2n+i])\big)\\
	&\vee \big(({\bf w}[j+2n+i]=\diamond) \wedge ({\bf w}[j+i]={\bf w}[j+n+i])\big)\\
	&\vee \big(({\bf w}[j+i]=\diamond) \wedge ({\bf w}[j+n+i]=\diamond)\big) \\
	&\vee \big(({\bf w}[j+i]=\diamond) \wedge ({\bf w}[j+2n+i]=\diamond)\big) \\
	&\vee \big(({\bf w}[j+n+i]=\diamond) \wedge ({\bf w}[j+2n+i]=\diamond)\big)\Big)
\end{align*}
where we find a more drastic difference compared to what would be the first-order logic for the case of full words.
We note that in the logical expression for a cube we would only need 
\[({\bf w}[j+i] = {\bf w}[j+n+i]) \wedge ({\bf w}[j+n+i] = {\bf w}[j+2n+i])\]
in the latter half of the implication for full words.
One could continue to consider higher powers, and the number of additional clauses in the partial word version will continue to grow.
We will restrict our attention to squares and cubes.

\begin{table}
\centering
\begin{tabular}{|c |c |c |c |c| c|}\hline
$\forall$ & $\exists$ & $\wedge$ & $\vee$ & $\neg$ & $\implies$ \\
\texttt{A} & \texttt{E} & \texttt{\&} & \texttt{|} & \texttt{\~} & \texttt{=>}\\ \hline
\end{tabular}
\caption{Logical operators and corresponding symbols in Walnut.}
\label{tbl:logic}
\end{table}

The proofs of many results in this paper are given by short snippets of Walnut code.
We now explain a few aspects of Walnut to make these snippets more readable to a reader that does not have prior experience with this language.
A more detailed explanation can be found in~\cite{Walnut}.
One can see in Table~\ref{tbl:logic} how usual logical symbols are represented in Walnut.
Our alphabet will always be $\{0,1,\dots,N\}$ for some $N$, and we will use $N+1$ to denote the hole $\diamond$.
For example, the binary partial word $0\diamond 1 0$ in Walnut would be \texttt{0210}.
The symbol \texttt{@} is use in denote a character of the alphabet as oppose of an integer which can be the index of a position in a word.
So, \texttt{W[2] = @3} is used to say that the character $3$ is in position $2$ of the word $W$.
Lastly, morphisms can be defined intuitively where \texttt{0->010} encodes $0 \mapsto 010$ the image of $0$.

\section{Results}

\subsection{Avoiding long squares in binary}\label{sec:binary}
In this subsection we look at binary partial words which only contain short squares and antisquares.
It is not possible to completely avoid squares since any binary (full) word of length at least $4$ will contain some square.
Let the morphism $h: \{0,1,2\} \to \{0,1,\diamond\}$ be defined by 
\begin{align*}
h(0) &= 1100\\
h(1) &= 011 \diamond \\
h(2) &= 1010\\
\end{align*}
which is a partial word variant of a morphism from~\cite[Section 2]{3squares} that itself is a variant of a morphism used in~\cite{EJS} to produce a binary word which has no squares of order $3$ or more.
In particular, if the hole in the definition of the morphism is replaced with a $1$, then the image under $h$ of any ternary square-free word will be a binary word where all squares have order less than $3$.
With the aid of Walnut one can quickly experiment, by replacing latter in the morphism with hole, and obtain a partial word which avoids squares of large length.
After the addition of the hole we no longer avoid squares of order $3$, but the resulting partial word will avoid squares of order $4$ or more.

\begin{theorem}
The partial word $h({\bf vtm})$ is a binary partial word with infinitely many holes that avoids squares of order $4$ or more.
\label{thm:longsq}
\end{theorem}
\begin{proof}
The word ${\bf vtm}$ is in Walnut as \texttt{VTM}.
So, all we need to do is define the morphism $h$ and apply it to ${\bf vtm}$ and check for squares of order $4$ or more.
In the code below \texttt{Wh} denotes the image of this morphism.
Recall, Walnut only uses $0,1,\dots,$ as letters.
So, the image of $h$ is a binary partial word represented over $\{0,1,2\}$ where $2$ plays the role of $\diamond$.
Running the following in Walnut
\begin{verbatim}
morphism h "0->1100, 1->0112, 2->1010";
image Wh h VTM;
eval no_sq "?msd_2 ~ Ej En Ai (n>3) & ((i<n)=>((Wh[j+i]=Wh[j+n+i]) 
        | Wh[i+j]=@2 | Wh[i+n+j]=@2))";
\end{verbatim}
returns ``TRUE'' which proves the result.
\end{proof}

To our knowledge the construction in Theorem~\ref{thm:longsq} is new in the context of partial words, but the result on avoiding long squares is not optimal.
In~\cite[Theorem 4]{largeSq} an infinite binary partial word with infinitely many holes is constructed so that the only squares compatible with factors of it are $0^2$, $1^2$, $(01)^2$, and $(11)^2$.
This is based on a construction for full words from~\cite{RWS}.
We can reproduce this optimal construction after defining the following two morphisms.

\begin{align*}
0 &\stackrel{\psi}{\to} 012321012340121012321234 & 0 &\stackrel{\phi}{\to} \diamond 11100\\
1 &\stackrel{\psi}{\to} 012101234323401234321234 & 1 &\stackrel{\phi}{\to} 101100 \\
2 &\stackrel{\psi}{\to} 012101232123401232101234 & 2 &\stackrel{\phi}{\to} 111000 \\
3 &\stackrel{\psi}{\to} 012321234323401232101234 & 3 &\stackrel{\phi}{\to} 110010\\
4 &\stackrel{\psi}{\to} 012321234012101234321234 & 4 &\stackrel{\phi}{\to} 110001
\end{align*}

\begin{theorem}[{\cite[Theorem 4]{largeSq}}]
The only full squares compatible with the factors of the binary partial word  $\phi(\psi^{\omega}(0))$  are $00$, $11$, $0101$, and $1111$.
\label{thm:400G}
\end{theorem}
\begin{proof}
This theorem can be proven using Walnut to first verify\footnote{The computation takes approximately 400G of storage due to the size of automata arising in the calculation.} the partial word has no squares are order greater than $2$. 
The presence or absence of each possible square are order $1$ or $2$ can then be checked either by hand or with Walnut.
We suppress the Walnut code as the $24$-uniform morphism $\psi$ does not display well succinctly.
\end{proof}

We next move the considering both squares and antisquares.
Consider the morphism $f:\{0,1,\dots,7\}^* \to \{0,1,\dots,7\}^*$ given by
\begin{align*}
f(0) &= 01 & f(1) &= 23\\
f(2) &= 24 & f(3) &= 51\\
f(4) &= 06 & f(5) &= 01 \\
f(6) &= 74 & f(7) &= 24
\end{align*}
along with the coding $g:\{0,1,\dots,7\}^* \to \{0,1,\diamond\}$ by $g(m) = m \pmod 2$ for $m \neq 6$ and $g(6) = \diamond$.
Applying $g$ to the fixed point of $f$ will give us a word avoiding both squares and antisquares of large length.
Since $f(0) = 01$ we may iterate applying $f$ to $0$ to obtain the unique fixed point of the morphism $f$ we denote by $f^{\omega}(0)$.
We will make use of the notation $f^{\omega}$ to denote fixed points of morphisms elsewhere as well.

\begin{theorem}
The partial word $g(f^{\omega}(0))$ is a binary partial word with infinitely many holes that avoids squares of order $7$ or more and avoids antisquares of order $3$ or more.
\label{thm:anti}
\end{theorem}

\begin{proof}
We establish the theorem by running the following in Walnut
\begin{verbatim}
morphism f "0->01, 1->23, 2->24, 3->51, 4->06, 5->01, 6->74, 7->24";
morphism g "0->0, 1->1, 2->0, 3->1, 4->0, 5->1, 6->2, 7->1";
promote Wf f;
image Wg g Wf;
eval no_sq "?msd_2 ~ Ej En Ai (n>6) & ((i<n)=>((Wg[j+i]=Wg[j+n+i]) 
        | Wg[i+j]=@2 | Wg[i+n+j]=@2))";

eval no_anti "?msd_2 ~ Ej En Ai (n>2) & ((i<n)=>((Wg[j+i]!=Wg[j+n+i]) 
        & Wg[i+j]!=@2 & Wg[i+n+j]!=@2))";
\end{verbatim}
which outputs ``TRUE'' twice.
\end{proof}

\begin{remark}
In was shown in~\cite[Theorem 9]{anti} that the full word obtained by coding the fixing point of $f$ with $m \mapsto m \pmod 2$ for all $m \in  \{0,1,\dots,7\}$ avoids both squares and antisquares of order $3$ or more.
Since adding holes cannot make any antisquares the fact about avoiding antisquares in Theorem~\ref{thm:anti} is immediate.
\end{remark}

\begin{table}
\centering
\begin{tabular}{|c| c c c c c c c c c c c c c  c c |}\hline
\backslashbox{$a$}{$b$} & 0 & 1 & 2 & 3 & 4 & 5 & 6 & 7 & 8 & 9 & 10 & 11 &12 & 13 & $\cdots$ \\ \hline
0&1&1&1&1&1&1&1&1&1&1&1&1&1&1& $\cdots$ \\ 
1&3&4&5&5&5&5&5&5&5&5&5&5&5&5 & $\cdots$ \\ 
2&5&7&9&10&11&11&12&12&14&14&15&16&16&16 & $\cdots$ \\
3&7&11&14&19&19&19&19&22&26&30&34&52&97&&\\ 
4&9&15&22&27&30&45&54&103&397&& & & & &\\ 
5&11&19&35&40&74& & &&&&&&&&\\ 
6&13&23&47&50&  &  &&&&&&&&&\\
7&15&27&59& &&&&&&&&&&&\\
8&17&31&147& &&&&&&&&&&& \\ 
9&19&35& &&&&&&&&&&&&\\
10&21&39& &&&&&&&&&&&& \\ 
$\vdots$ & $\vdots$ & $\vdots$& &&&&&&&&&&&&\\ \hline
\end{tabular}
\caption{The length of the longest binary partial word with a single hole that contains at most $a$ squares and at most $b$ antisquares.}
\label{tbl:anti}
\end{table}

In~\cite{anti} for any fixed $a$ and $b$, the problem of finding the longest binary full word which contains at most $a$ distinct squares and at most $b$ distinct antisquares was solved.
One can consider versions of the same problem for partial words.
For example, we can ask for the length of the longest binary partial word with a fixed number of holes that contains at most $a$ distinct squares and at most $b$ distinct antisquares.
We will focus on the case of partial words with a single hole.
We will count how many distinct squares are compatible with a partial word.
For example, $0\diamond 0$ contain only the one square $0^2$ which is compatible with both $\diamond 0$ and $0 \diamond$.
Counting distinct squares has received much attention for both words and partial words.
It is known that even the addition of a single hole can fundamentally change distinct squares~\cite{reu, BSMS, HHK, IPL}.
Unlike if we were simply avoiding squares, the length of this longest binary partial word can be shorter or longer than that of the corresponding full word.
This is since replacing a letter in a binary full word can possibly create a square while it also has the potential to remove an antisquare.

Consider the following example for $a=4$ and $b=5$.
The length $32$ partial word
\[\diamond0111010000011010000110000010000\]
contains only the squares $0^2$, $1^2$, $(00)^2$,  and $(10)^2$.
It also contains only the antisquares $01$, $10$, $0011$, $0110$, and $1100$.
For full words the optimal length is $31$ which was originally computed in~\cite[Figure 1]{anti}. One binary full word giving witness to this length is 
\[0111010000011010000110000010000\]
which is obtained from the partial word above by removing the hole.
Notice prepending $0$ to this full word creates the new antisquare $001110$ while prepending $1$ creates $1011101000$.
Prepending $\diamond$ allows us to avoid new antisquares.
Moreover, we can create even longer partial words with one hole and only $4$ distinct squares and $5$ distinct antisquares.
The optimal length is $45$.
The partial word
\[000010000011000010110000011\diamond 00101100000101110\]
has length $45$ and only contains the squares $0^2$, $1^2$, $(00)^2$, and $(01)^2$ as well as only the antisquares $01$, $10$, $0011$, $0110$, and $1100$.
In Table~\ref{tbl:anti} we list the optimal lengths for many values of $a$ and $b$.

\begin{theorem}
The values in Table~\ref{tbl:anti} are correct.
\end{theorem}
\begin{proof}
The first three rows and first two columns are each infinite sequences of finite values and must be proven.
Outside of these rows and columns there are only finitely many entries which can be computed but a backtracking approach by generating partial words with at most one hole and checking if they can be extended.
The first three rows follow from the fact that a binary partial word with one hole that contains at most $0$, $1$, or $2$ squares has length at most $1$, $5$, or $16$ respectively.
This can be seen by direct verification of this fact, then computing the values until we reach $1$, $5$ or $16$.
Note it is clear the entries weakly increase along row and column.

For the first column, up to complementation, we consider words of the form $0^m\diamond 0^n$ or $0^m\diamond 1^n$ which are the only binary partial words with a single hole that do not contain an antisquare.
Let $\ell = m+n+1$ be the length of such a word.
The number of squares such a word contains is
\[\left\lfloor \frac{m+1}{2} \right\rfloor + \left\lfloor \frac{n+1}{2} \right\rfloor\ = \left \lfloor \frac{\ell}{2} \right\rfloor\]
and so $\ell = 2a + 1$ is the longest length containing at most $a$ distinct squares and $0$ antisquares.

Now for the second column we consider partial words with a single hole and only $1$ antisquare.
Let us assume $a > 1$.
We will look at partial words which start with $0$ and contain only the antisquare $01$.
So, we have  $0^m 1^n \diamond 0^p 1^q$ with $m > 0$.
Note if $m > 1$, then $n \leq 1$ or else we have both the antisquares $01$ and $0011$.
Similarly if $p > 1$, then $q \leq 1$.
So, let us assume our partial word is $0^m 1 \diamond 0^p 1$.
This partial word has $\left\lfloor \frac{N}{2} \right\rfloor + 1$ distinct squares where $N = \max(m, p+1)$ unless $m = p+1$ and the whole partial word then gives an additional square.
The squares contained are $1^2$ and $0^{2k}$ for $2k < N$ along with possibly $(0^m1)^2$.
So, we may take $0^{2a-2} 1 \diamond 0^{2a-2}1$ which has length $4a - 1$ and contains $a$ squares and $1$ antisquare.
This gives us one such word realizing the maximum agrees with what is found in Table~\ref{tbl:anti}.
We also have the partial word $0^{2a-1} 1 \diamond 0^{2a-2}$ of length $4a-1$ which contains $a$ squares and $1$ antisquare.
It turns out that all other such partial words can be obtained from the two we have given by complement and reversal.
This can be checked by considering the remaining cases of the possible forms of the partial word in a similar manner.
\end{proof}

\begin{remark}
Table~\ref{tbl:anti} is incomplete, and we do not know if the missing entries are finite or infinite.
From~\cite{anti} it is known that for full binary words the entry corresponding to $a=5$ and $b=5$ is finite while the remaining missing entries are infinite.
\end{remark}

\subsection{Avoiding non-trivial squares and cubes with many holes}
In this subsection we demonstrate how  some known constructions~\cite{denseLATA} of partial words ``dense'' with holes avoiding powers can be obtained and verified through Walnut.
The \emph{hole sparsity} of a partial word is smallest $s$ such that every factor of length $s$ contains at least one hole.
A square of the form $a\diamond$ or $\diamond a$ for some letter $a$ is called a \emph{trivial square}.
Indeed every partial word of length at least $2$ which contains at least one hole as well as at least one letter will contain a trivial square.
Thus for avoidance purposes we must allow trivial squares and avoid non-trivial squares.
Since the presence of holes makes squares or cubes more likely, it is an interesting problem to find partial words with small hole sparsity (and hence many holes) which avoid non-trivial squares or cubes (or more generally higher powers).

Let us consider the morphism $\rho: \{0,1,2, 3\}^* \to \{0,1,2,3\}^*$ defined by
\begin{align*}
\rho(0) &= 03\\
\rho(1) &= 12\\
\rho(2) &= 01\\
\rho(3) &= 10
\end{align*}
which appears in~\cite[Exercise 33(c)]{auto}.
Next we consider the morphism $\sigma:  \{0,1,2, 3\}^* \to \{0,1,2,3,\diamond\}^*$ defined by
\begin{align*}
\sigma(0) &= 320\diamond\\
\sigma(1) &= 120\diamond\\
\sigma(2) &= 310\diamond\\
\sigma(3) &= 130\diamond
\end{align*}
whose image is a partial word over an alphabet with size $4$.
We are now ready to give on automated proof of the following which is in the proof of~\cite[Lemma 2]{denseLATA}.

\begin{theorem}
The partial word $\sigma(\rho^{\omega}(0))$ is a partial word with hole sparsity $4$ over an alphabet of size $4$ which avoids non-trivial squares.
\label{thm:dense_sq}
\end{theorem}
\begin{proof}
It is easy to see that the squares $\diamond^2$, $0^2$, $1^2$, $2^2$, and $3^2$ do not occur in $\sigma(\rho^{\omega}(0))$.
Recall, since we avoiding non-trivial squares $a \diamond$ and $\diamond a$ are allowed to we present for $a \in \{0,1,2,3\}$.
Thus, we only need to worry about squares of order $n$ for $n>1$.
Running the following commands in Walnut
\begin{verbatim}
morphism rho "0->03, 1->12, 2->01, 3->10";
morphism sigma "0->3204, 1->1204, 2->3104, 3->1304";
promote Wrho rho;
image W sigma Wrho;
eval no_sq "?msd_2 ~ Ej En Ai (n>1) & ((i<n)=>((W[j+i]=W[j+n+i]) 
        | W[i+j]=@4 | W[i+n+j]=@4))";
\end{verbatim}
results in an output of ``TRUE'' and the theorem is proven.
\end{proof}

We let $\tau$ denote the morphism defined by $\tau(0) = 01\diamond$ and $\tau(1) = 02\diamond$. We can now give an automated proof of the following which was first proven in~\cite[Lemma 7]{denseLATA}.

\begin{theorem}
The partial word $\tau({\bf tm})$ is partial word with hole sparsity $3$ over an alphabet of size $3$ which avoids cubes.
\label{thm:dense_cube}
\end{theorem}

\begin{proof}
The word ${\bf tm}$ is contained in Walnut at \texttt{T}.
We run the following in Walnut
\begin{verbatim}
morphism tau "0->013, 1->023";
image W tau T;
eval no_cube "?msd_2 ~Ej En Ai (n>1)&((i<n)=>(
     W[j+i]=W[j+n+i] & W[j+n+i]=W[j+2*n+i])
     |(W[j+i]=@3 & W[j+n+i]=W[j+2*n+i])|(W[j+n+i]=@3 & W[j+i]=W[j+2*n+i]) 
     |(W[j+2*n+i]=@3 & W[j+i]=W[j+n+i])|(W[j+i]=@3 & W[j+n+i]=@3) 
     |(W[j+i]=@3 & W[j+2*n+i]=@3) | (W[j+n+i]=@3 & W[j+2*n+i]=@3))";
\end{verbatim}
and it outputs ``TRUE'' proving the theorem.
\end{proof}

\begin{remark}
Both Theorem~\ref{thm:dense_sq} and Theorem~\ref{thm:dense_cube} are optimal in establishing smallest hole sparsity avoiding non-trivial squares and cubes respectively for a given alphabet size~\cite{JALC}.
For example, there does not exists an infinite partial word with hole sparsity $3$ over a $4$ letter alphabet which avoids non-trivial squares.
\end{remark}

\section{Some alternative notions}
Many times there are multiple ways to consider a full word concept in terms of partial words.
Here we consider some alternatives to definitions we used earlier.
We first adjust the idea of an antisquare in a partial word.
After we look at a different way to view cubes in partial words with periodicity.

\subsection{Another take on antisquares}
Let us say for a partial word ${\bf w}$ contains a \emph{c-antisquare} whenever
\[\exists j  \exists (n > 0)  \forall i, (i < n) \implies \left(({\bf w}[j+i] \neq {\bf w}[j+n+i]) \vee ({\bf w}[j+i]=\diamond) \vee ({\bf w}[j+n+i]=\diamond)\right)\]
which exactly means there is a full word antisquare which is compatible.
Now the presence of a hole makes is easier for partial words to contain both squares as well as antisquares.

As before we can consider the problem of find the long partial word with a single hole that contains at most $a$ distinct squares and $b$ distinct antisquares.
Again we count distinct full word squares and antisquares compatible with factors of our partial words.
For example, $0\diamond 0$ has both the antisquares $01$ and $10$ compatible with $0 \diamond$ and $\diamond 0$ respectively.
However, for $0 \diamond 1$ we only count a single c-antisquare $01$ compatible with both $0 \diamond$ and $\diamond 1$.
Note neither $0 \diamond 0$ nor $0 \diamond 1$ would have any antisquares as partial words in the setting of Section~\ref{sec:binary}.

\begin{table}
\centering
\begin{tabular}{|c| c c c c c c c c c c c c c  c c |}\hline
\backslashbox{$a$}{$b$} & 0 & 1 & 2 & 3 & 4 & 5 & 6 & 7 & 8 & 9 & 10 & 11 &12 & 13 & $\cdots$ \\ \hline
0&1&1&1&1&1&1&1&1&1&1&1&1&1&1& $\cdots$ \\ 
1&1&3&4&5&5&5&5&5&5&5&5&5&5&5& $\cdots$ \\ 
2&1&5&9&9&9&9&10&12&12&13&13&16&16&16& $\cdots$ \\
3&1&7&13&13&13&17&17&17&22&27&32&52&&\\ 
4&1&9&18&18&25&29&44&55&&& & & & &\\ 
5&1&11&24&24&37& & &&&&&&&&\\ 
6&1&13&30&36&  &  &&&&&&&&&\\
7&1&15&43& &&&&&&&&&&&\\
8&1&17&147& &&&&&&&&&&& \\ 
9&1&19& &&&&&&&&&&&&\\
10&1&21& &&&&&&&&&&&& \\ 
$\vdots$ & $\vdots$ & $\vdots$& &&&&&&&&&&&&\\ \hline
\end{tabular}
\caption{The length of the longest binary partial word with a single hole that contains at most $a$ squares and at most $b$ c-antisquares.}
\label{tbl:c-anti}
\end{table}

\begin{theorem}
The values in Table~\ref{tbl:c-anti} are correct.
\end{theorem}
\begin{proof}
For c-antisquares first three rows and first two columns are each infinite sequences of finite values and must be proven just as in the case of antisquares.
The other entries can be computed but a backtracking approach by extending partial words until they no longer meets the conditions imposed.
The first three rows follow from the fact that a binary partial word with one hole that contains at most $0$, $1$, or $2$ squares has length at most $1$, $5$, or $16$ respectively.
This can be seen by direct verification of this fact, then computing the values until we reach $1$, $5$ or $16$.
Note it is clear the entries weakly increase along row and column.

Each entry in the first column must be $1$ since either $a \diamond$ or $\diamond a$ is a c-antisquare for any letter $a$.
This leaves only the second column remaining.
Any partial word with a single hole of length great than $1$ contains at least one c-antisquare which is not an antisquare.
Also, any antisquare is also a c-antisquare.
So, for the second column we are looking for partial words with no antisquares and at most $a$ distinct squares.
This column turns out to exactly match the first column of Table~\ref{tbl:anti} and we can use the partial word $0^m \diamond 1^n$ as we did in the proof of Theorem~\ref{thm:anti}.
Here $01$ is the only antisquare compatible with any factor.
\end{proof}

\begin{remark}
In Table~\ref{tbl:anti} and Table~\ref{tbl:c-anti} we are counting squares in the same way.
Also, any time an antisquare is counted in the setting of Table~\ref{tbl:anti} a c-antisquare is also counted in the setting of Table~\ref{tbl:c-anti}.
So, it follows that values in Table~\ref{tbl:c-anti} are necessarily less than or equal to the corresponding entries in Table~\ref{tbl:anti}.
\end{remark}

\subsection{Weak overlaps}
A idea closely related to squares and cubes is an \emph{overlap} which in full words is a word of the form $axaxa$ for a word $x$ and letter $a$.
A \emph{weak overlap} is a partial word of the form $a_0 x_1 a_1 x_2 a_2$ such that $x_1 \uparrow x_2$, $a_0 \uparrow a_1$, and $a_1 \uparrow a_2$.
In the case that $|x_1| = 0 = |x_2|$ we additionally require that $a_0 \uparrow a_2$.
.Notice without the addition condition when $|x_1| = 0 = |x_2|$ it would be the case that $a \diamond b$ is a weak overlap for any letters $a$ and $b$.

It need not be the case the a weak overlap is contained in a full word overlap.
For example, ${\bf w} = 01 \diamond 11$ is a weak overlap.
This is an example of a \emph{weak period} where ${\bf w}[i] \uparrow {\bf w}[i+p]$ with $p=2$ in this case.
A \emph{strong period} is a partial word with ${\bf w}[i] \uparrow {\bf w}[j]$ whenever $i \equiv j \pmod{p}$.
Our example of a weak overlap does not have a strong period.
Looking back to our definition of cubes in Section~\ref{sec:defs} we see that cubes are a strong period.
Presently we will focus on weak periods by considering weak overlaps.
Since it is easier for weak periods to occur, it is thus more difficult to avoid them.

We define a morphism $\gamma \{0,1,2,3\}^* \to \{0,1,2,3\}^*$ by
\begin{align*}
\gamma(0) &= 03\\
\gamma(1) &= 02\\
\gamma(2) &= 21\\
\gamma(3) &= 20
\end{align*}
and also
\begin{align*}
\delta(0) = 01302\\
\delta(1) = 01234\\
\delta(2) = 43142\\
\delta(3) = 43210
\end{align*}
another morphism $\delta: \{0,1,2,3\}^* \to \{0,1,2,3,4\}^*$.
Lastly, for each $0 \leq i \leq 4$ define $\phi_i: \{0,1,2,3,4\}^* \to \{0,1,2,3,4,\diamond\}^*$ by $\phi_i(i) = \diamond$ and $\phi_i(j) = j$ when $j \neq i$.
The morphism $\phi_i$ simply substitutes $\diamond$ for $i$ and acts as the identity on the remaining letters.
The following theorem gives examples of weakly overlapfree partial words which are cases of partial words in~\cite[Theorem 5]{olap}.

\begin{theorem}
For any $0 \leq i \leq 4$ the partial word $\phi_i(\delta(\gamma^{\omega}(0)))$ is weakly overlapfree.
\end{theorem}
\begin{proof}
For $ i = 0$ we run the following in Walnut
\begin{verbatim}
morphism g "0->03, 1->02, 2->21, 3->20";
promote w g;

morphism d "0->01302, 1->01234, 2->43142, 3->43210";
image v d w;

eval ov0 "~Ei En Aj (n > 1) & ((j <= n) => (v[i+j] = v[i+j+n] 
	| v[i+j]=@0 | v[i+j+n]=@0))";

eval t0 "~Ei (v[i]=@0 & v[i+1]=@0)|(v[i]=@0 & v[i+2]=@0)|(v[i]=@0 & v[i+1]=v[i+2]) 
	|(v[i+1]=@0 & v[i]=v[i+1])|(v[i+2]=@0 & v[i]=v[i+1])";
\end{verbatim}
which outputs ``TRUE'' twice.
The  statment \texttt{eval ov0} checks for weak overlaps $a_0x_1a_1x_2a_2$ with $|x_1|, |x_1| > 0$, and \texttt{eval t0} checks the case that $|x_1| = 0 = |x_2|$.
The cases with $i > 0$ are similar.
\end{proof}

\section{Conclusion}

We have initiated a study of partial words paired with the theorem prover Walnut.
Additionally, we have extended the definition of antisquares to partial words.
We believe both directions could be a source of new problems in combinatorics on words.
Furthermore, we have given alternative proofs of some results on partial words which provide machine verification.
We have discussed how the logical statements expressing a square are longer for partial words than full words.
So this adds some complexity to the Walnut calculations.
Additionally, in Walnut partial words work with an alphabet with an extra letter which represents the hole.
Let us close with an example comparing partial words with full words in Walnut.

Let us consider the morphisms $g$ and $h$ defined by
\begin{align*}
g(0) &= 1100 & h(0) &= 1100\\
g(1) &= 0111 & h(1) &= 011 \diamond \\
g(2) &= 1010 & h(2) &= 1010
\end{align*}
where $h$ was the morphism used in Theorem~\ref{thm:longsq}.
To get an idea of what happens going from full words to partial words we have the deterministic finite automata with output (DFAO) for $g({\bf vtm})$ and $h({\bf vtm})$ shown in Figure~\ref{fig:g} and Figure~\ref{fig:h} respectively.
Walnut works with automatic sequences using their DFAOs.
We find in this case only a modest increase in the size of the DFAO, and we were indeed able to use Walnut to automate a proof in the more complex but still tractable partial word situation.

\begin{figure}
\centering
\includegraphics[scale=0.42]{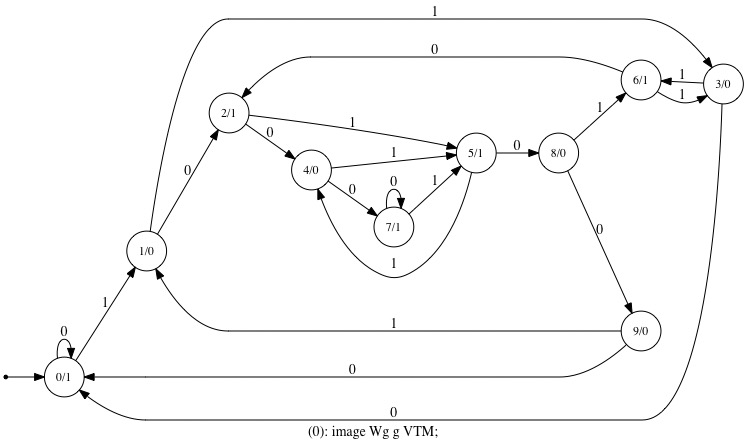}
\caption{The DFAO for $g({\bf vtm})$.}
\label{fig:g}
\end{figure}

\begin{figure}
\centering
\includegraphics[scale=0.42]{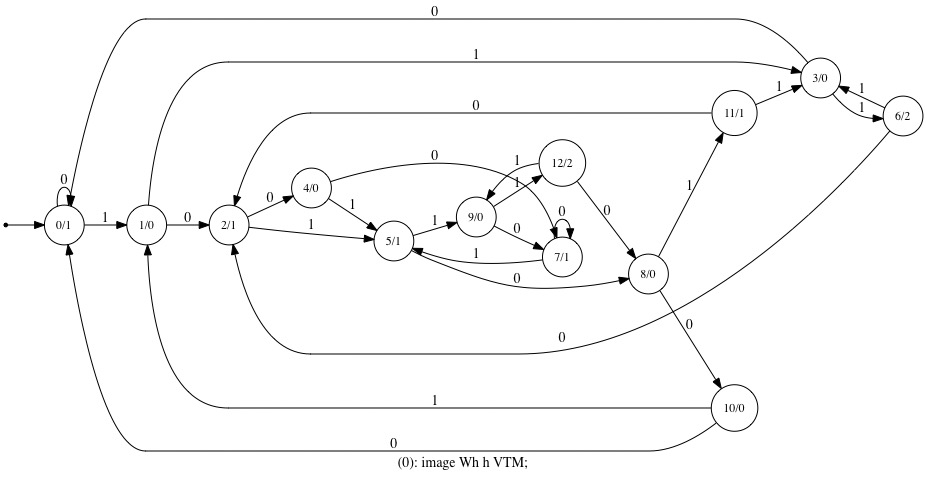}
\caption{The DFAO for $h({\bf vtm})$.}
\label{fig:h}
\end{figure}

To show there are no squares of length greater than $3$ in neither $g({\bf vtm})$ nor $h({\bf vtm})$ we may run
\begin{verbatim}
morphism h "0->1100, 1->0112, 2->1010";
image Wh h VTM;
eval no_sq "?msd_2 ~ Ej En Ai (n>3) & ((i<n)=>((Wh[j+i]=Wh[j+n+i]) 
        | Wh[i+j]=@2 | Wh[i+n+j]=@2))";

morphism g "0->1100, 1->0111, 2->1010";
image Wg g VTM;
eval no_sq_full "?msd_2 ~ Ej En Ai (n>3) & ((i<n)=>(Wg[j+i]=Wg[j+n+i]))";
\end{verbatim}
in Walnut.
In the log files produced we will find the following two lines
\begin{small}
\begin{verbatim}
(i<n=>((Wh[(j+i)]=Wh[((j+n)+i)]|Wh[(i+j)]=@2)|Wh[((i+n)+j)]=@2)):229 states - 24ms
\end{verbatim}
\end{small}
and
\begin{small}
\begin{verbatim}
(i<n=>Wg[(j+i)]=Wg[((j+n)+i)]):217 states - 10ms
\end{verbatim}
\end{small}
showing the sizes of the automata Walnut needs to determine the exists of a square in either $g({\bf vtm})$ or $h({\bf vtm})$ respectively.
We see the partial word version requires 229 states more states and takes to build 24ms compared 217 states and to 10ms for the full word version.
The issue one encounters in Walnut computations is typically in issue of space due to building some automaton.
This example, and the others we have given, suggest that partial word versions of theorems may take slightly more space in Walnut but may often tractable with Walnut when their full word counterparts are.

\bibliography{partial_ref.bib}

\end{document}